\documentclass[prd,nofootinbib,superscriptaddress,twocolumn]{revtex4}

\usepackage{amsfonts,amssymb,amsthm,bbm}

\usepackage{amsmath}

\usepackage{hyperref}

\usepackage{color,psfrag}
\usepackage[dvips]{graphicx}

\usepackage{tikz}
\usetikzlibrary{calc}
\usetikzlibrary{decorations.pathmorphing}
\usetikzlibrary{shapes.geometric}
\usetikzlibrary{arrows,decorations.markings}

\newcommand{\C}{{\mathbb C}}
\newcommand{\N}{{\mathbb N}}

\newcommand{\cH}{{\mathcal H}}

\newcommand{\cP}{{\mathcal P}}

\newcommand{\cI}{{\mathcal I}}

\newcommand{\SU}{\mathrm{SU}}

\newcommand{\SL}{\mathrm{SL}}

\newcommand{\be}{\begin{equation}}
\newcommand{\ee}{\end{equation}}
\newcommand{\beq}{\begin{eqnarray}}
\newcommand{\eeq}{\end{eqnarray}}
\newcommand{\bes}{\begin{eqnarray}}
\newcommand{\ees}{\end{eqnarray}}

\newcommand{\alink}[4]
{\draw[decoration={markings,mark=at position 0.6 with {\arrow[scale=1.5,>=stealth]{>}}},postaction={decorate}] (#1) -- node[#3,pos=.5]{$#4$}(#2)}
\newcommand{\link}[2]
{\draw[decoration={markings,mark=at position 0.6 with {\arrow[scale=1.5,>=stealth]{>}}},postaction={decorate}] (#1) --(#2)}

\newcommand{\su}{{\mathfrak{su}}}

\newcommand{\la}{\langle}
\newcommand{\ra}{\rangle}

\newcommand{\tr}{{\mathrm{Tr}}}
\newcommand{\f}{\frac}

\def\nn{\nonumber}
\def\pp{\partial}

\def\rd{\mathrm{d}}

\def\vphi{\varphi}

\newcommand{\id}{\mathbb{I}}

\def\vJ{\vec{J}}
\def\trho{\widetilde{\rho}}
\def\tpsi{{\Psi}}
\def\tE{\widetilde{E}}
\def\tk{\tilde{k}}
\def\tm{\tilde{m}}

\newtheorem{theorem}{Theorem}[section]
\newtheorem{lemma}[theorem]{Lemma}
\newtheorem{prop}[theorem]{Proposition}
\newtheorem{coro}[theorem]{Corollary}

\begin{document}

\title{Intertwiner Entanglement on Spin Networks}

\author{{\bf Etera R. Livine}}\email{etera.livine@ens-lyon.fr}
\affiliation{Universit\'e de Lyon, ENS de Lyon, CNRS, Laboratoire de Physique LPENSL, 69007 Lyon, France}

\date{\today}

\begin{abstract}

In the context of quantum gravity, we clarify entanglement calculations on spin networks: we distinguish the gauge-invariant entanglement between intertwiners located at the nodes and the  entanglement between spin states located on the network's links. We compute explicitly these two notions of entanglement between neighboring nodes and show that they are always related to the typical $\ln(2j+1)$ term depending on the spin $j$ living on the link between them. This $\ln(2j+1)$ contribution comes from  looking at non-gauge invariant states, thus we interpret it as gauge-breaking and unphysical. In particular, this confirms that pure spin network basis states do not carry any physical entanglement, so that true entanglement and correlations in loop quantum gravity comes from spin or intertwiner superpositions.

\end{abstract}

\maketitle


In the background independent framework of quantum gravity, quantum states of geometry are defined up to diffeomorphisms with no reference to any coordinate systems or special choice of metric. Then distances and curvature, which define the geometry, are reconstructed from the interaction between subsystems, best quantified by the correlation and entanglement between those subsystems. This perspective sets the field of quantum information at the heart of research in quantum gravity, with essential roles to play for entanglement, decoherence and quantum localization in probing quantum states of geometries and thinking about the quantum-to-classical transition for the space-time geometry.

The loop quantum gravity framework provides a local definition of quantum states of space geometry as spin networks. They are collections of quantum states invariant under 3d rotations linked to each other by rotations encoding the change of frame from one state to the next, thus forming a network of frame transformations defining the quantum 3d space. These spin networks can further be understood as the quantization of discrete ``twisted'' geometries \cite{Freidel:2010aq,Dupuis:2012vp,Dupuis:2012yw}, where the 3d space is made of quantized flat chunks of volume glued together. Then the spin networks' evolution is described by spinfoam path integrals, built from topological quantum field theory (TQFT) \cite{Barrett:1997gw,Baez:1997zt,Livine:2010zx,Perez:2012wv}. This quantum geometry concept is ultimately meant to replace the paradigm of classical smooth manifold  to describe the space-time geometry.

Spin networks describe the geometry at the Planck scale. They need to be coarse-grain in order to derive the structure of quantum geometries at larger scales. The issue is that the definition of  spin networks is ultra-local, with algebraic structures associated to each node of the network without any overall constraint (at the kinematical level), so it is a hard problem to get a dynamical coherent picture of smooth geometry arising semi-classically. The quantum gravity dynamics, implementing the quantum Einstein equations and the diffeomorphism invariance of quantum geometry states, should ultimately address this question. However we would still require good observables to probe the large scale structure of spin networks. The idea is to study correlations and entanglement on spin networks and use them to probe and reconstruct the geometry, as suggested in \cite{Livine:2006xk,Feller:2015yta,Bianchi:2016tmw,Chirco:2017xjb}. The goal would be to identify holographic states, satisfying a local area-entropy law, with algebraically-decaying correlations on mesoscopic scales allowing for a semi-classical interpretation as smooth geometries at some point between the Planck scale and our scale. This line of research follows the same logic than the holographic reconstruction of geometry which is being developed from the viewpoint of the AdS/CFT correspondence and gauge-gravity dualities.

\medskip

Here, we would like to start by clarifying the notion of entanglement on spin network states. We insist on respecting the local $\SU(2)$ gauge invariance and on using $\SU(2)$-invariant observables to probe spin networks.

A spin network state is constructed on a closed oriented graph $\Gamma$. The Hilbert space of wave-functions (for quantum states of geometry) on that graph is the space of $L^{2}$ functions of $\SU(2)$ group elements on every link $e\in\Gamma$ while being invariant under the $\SU(2)$ action at every node $v\in\Gamma$:
\be
\cH_{\Gamma}=L^{2}(\SU(2)^{\times E}/\SU(2)^{\times V})\,,
\ee
where $E$ is the number of edges in $\Gamma$ and $V$ its number of vertices. The $\SU(2)$ gauge invariance at the nodes of an arbitrary state $\vphi$ reads:
\be
\vphi(\{g_{e}\}_{e\in\Gamma})=\vphi(\{h_{s(e)}g_{e}h_{t(e)}^{-1}\}_{e\in\Gamma})
\quad
\forall h_{v}\in\SU(2)^{\times V}
\,,
\nn
\ee
where $s(e)$ stands for the source vertex of the oriented $e$ and $t(e)$ the corresponding target vertex. We define basis states on $\cH_{\Gamma}$ as {\it spin networks}, labeled by a spin $j_{e}\in\f\N2$ on each edge $e$ and an intertwiner $I_{v}$ at every node $v$, as illustrated on fig.\ref{fig:spinnetwork}:
\be
\cH_{\Gamma}=\bigoplus_{\{j_{e},I_{v}\}_{e,v\in\Gamma}}\C\,|\{j_{e},I_{v}\}\ra\,.
\ee
A spin $j$ defines a irreducible representation of the Lie group $\SU(2)$. The associated  Hilbert space, noted $V^{j}$, has dimension $d_{j}=(2j+1)$. Writing $J^{a}$ with $a=x,y,z$ for the $\su(2)$ Lie algebra generators, we use the usual basis of $V^{j}$ diagonalizing both the $\su(2)$ Casimir $\vJ^{2}\equiv J^{a}J^{a}$ and the generator $J^{z}$, and labeled by the spin $j$ and the magnetic momentum $m$ running by integer step from $-j$ to $+j$:
\be
V^{j}=\bigoplus_{-j\le m\le j}\C\,|j,m\ra\,,
\quad
\dim V^{j}=2j+1
\,,
\ee
\be
\vJ^{2}\,|j,m\ra=j(j+1)\,|j,m\ra\,,\quad
J^{z}\,|j,m\ra=m\,|j,m\ra
\,.
\nn
\ee
An intertwiner $I_{v}$ at the node $v$ is a $\SU(2)$-invariant state (a.k.a. a singlet state) in the tensor product of all the spins living on the edges linked to $v$:
\be
I_{v}\in \cI_{v}=\mathrm{Inv}_{\SU(2)}
\big{[}
\bigotimes_{e\ni v}V^{j_{e}}
\big{]}\,.
\ee
By definition of an irreducible representation, bivalent intertwiners only exist if the two spins $j_{1}=j_{2}$ are equal and are then unique. Trivalent intertwiners between three spins exists if and only if those spins satisfy triangular inequalities, $|j_{1}-j_{2}|\le j_{3}\le(j_{1}+j_{2}) $ and are then unique: they are given by the Clebsh-Gordan coefficients. From valence 4 onwards, the intertwiner space grows in dimension and they are multiple non-trivial intertwiner states.

\begin{figure}[h!]

\begin{tikzpicture}[scale=1.4]

\coordinate(a) at (0,0) ;
\coordinate(b) at (.5,1);
\coordinate(c) at (.9,-.1);
\coordinate(d) at (.3,-.8);
\coordinate(e) at (1.3,.7);
\coordinate(f) at (2.3,1.1);
\coordinate(g) at (2,.4);
\coordinate(h) at (2.7,.1);
\coordinate(i) at (2.1,-.2);
\coordinate(j) at (1.4,.-.5);

\draw (a) node {$\bullet$} node[left]{$I_{A}$};
\draw (b) node {$\bullet$}node[above]{$I_{B}$};
\draw (c) node {$\bullet$} node[right]{$I_{C}$};
\draw (d) node {$\bullet$} node[below]{$I_{D}$};
\draw (e) node {$\bullet$} node[above]{$I_{E}$};
\draw (f) node {$\bullet$};
\draw (g) node {$\bullet$};
\draw (h) node {$\bullet$};
\draw (i) node {$\bullet$};
\draw (j) node {$\bullet$};

\alink{a}{b}{left}{j_{1}};
\alink{a}{c}{above}{j_{2}};
\alink{a}{d}{left}{j_{3}};
\link{b}{c};
\link{c}{d};
\link{b}{e};
\link{e}{f};
\link{c}{e};
\link{e}{g};
\link{f}{g};
\link{f}{h};
\link{h}{g};
\link{h}{i};
\link{g}{i};
\link{j}{i};
\link{d}{j};
\link{c}{j};

				
\end{tikzpicture}

\caption{A spin network, on a closed oriented graph $\Gamma$, is a basis state is labeled with  spins $j_{e}\in\f N2$ on the graph edges and intertwiner states -singlet states- on the graph vertices.}
\label{fig:spinnetwork}
\end{figure}
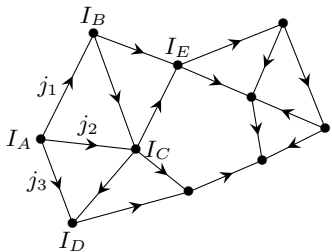

In the geometrical interpretation of loop quantum gravity, intertwiners represent the 3d volume excitations while spins give the quanta of area gluing neighboring chunks of volume (in other words, the cross-section). This is validated by the geometrical interpretation of intertwiners as quantized convex polyhedra \cite{Freidel:2009nu,Freidel:2009ck,Bianchi:2010gc,Livine:2013tsa} with the spins giving the area of the polyhedra' faces.

\medskip

We would like to study the basic case of two neighboring nodes, $A$ and $B$, linked by a single edge carrying a spin $j$, as depicted on fig.\ref{2vertex}.
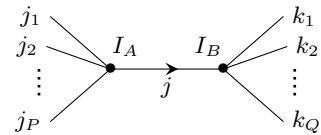
\begin{figure}[h]
\begin{center}

\begin{tikzpicture}[scale=1]

\coordinate(a) at (0,0) ;
\coordinate(b) at (1.5,0);

\draw (a) node {$\bullet$} ++(0.2,0.3) node{$I_{A}$};
\draw (b) node {$\bullet$} ++(-0.2,0.3)  node{$I_{B}$};
\alink{a}{b}{below}{j};

\draw (a)--++(-.85,0.3) node[left]{$j_{2}$};
\draw (a)--++(-.8,0.7) node[left]{$j_{1}$};
\draw (a)--++(-.8,-.7)  node[left]{$j_{P}$};

\draw (b)--++(.85,0.3) node[right]{$k_{2}$};
\draw (b)--++(.8,0.7) node[right]{$k_{1}$};
\draw (b)--++(.8,-.7)  node[right]{$k_{Q}$};

\draw[dotted,line width=1pt] (-.95,0) -- (-.95,-.4);
\draw[dotted,line width=1pt] (1.5+.95,0) -- (1.5+.95,-.4);

\end{tikzpicture}

\caption{Two vertices $A$ and $B$ on a spin network state, linked by a single edge carrying a spin $j$. The external spins are labeled $j_{1},..,j_{P}$ for the vertex $A$, while we label them $k_{1},..,k_{Q}$ for the vertex $B$.
We must distinguish the intertwiner Hilbert space  $\cH^0_{AB}=\cH^0_{A}\otimes \cH^0_{B}$, consisting in tensor products of intertwiner states at the nodes $A$ and $B$, from the boundary spin Hilbert space $H^{\pp}_{AB}=\bigotimes_{i}^P V^{j_{i}}\otimes \bigotimes_{i}^Q V^{k_{i}}$.
\label{2vertex}}
\end{center}
\end{figure}
This configuration is usually thought to produce an entanglement of $\ln(2j+1)$ between the two nodes \cite{Livine:2006xk,Chirco:2017xjb}. We will challenge this and argue that one must distinguish the entanglement between the two intertwiners and the entanglement between the spin states living on the external legs of $A$ and $B$.

We will show that the intertwiner entanglement is bounded by the intertwiner space dimensions and can be much larger than $\ln(2j+1)$. For a spin network basis state, the two intertwiners are completely decoupled and the intertwiner entanglement simply vanishes, while the spin state entanglement is always equal to $\ln(2j+1)$. This hints that this $\ln(2j+1)$ contribution is unphysical and is a signature of
looking at a non-gauge-invariant observable.
We further prove that for an arbitrary spin network superposition, the spin state entanglement is always equal to the intertwiner entanglement plus exactly $\ln(2j+1)$. This supports our claim that the intertwiner entanglement is the physical quantum correlation between nodes to consider in loop quantum gravity and that  true entanglement comes from spin network superpositions\footnotemark.
\footnotetext{
Since intertwiners can be thought of mathematically as spin labels on virtual links (within the nodes), this point of view is consistent with the spinfoam graviton calculations \cite{Bianchi:2006uf,Livine:2006it,Christensen:2007rv} and the recent correlation calculations on spin network states from \cite{Feller:2015yta} and \cite{Bianchi:2016tmw,Bianchi:2016hmk} where one looks at the correlation and entanglement between the spin labels on different labels (or between holonomy operator insertions).
}
%

\section{Intertwiner entanglement vs Spin state entanglement}

We focus on two neighboring nodes of a spin network, $A$ and $B$, as on fig.\ref{2vertex}, linked by a single edge decorated with a fixed given spin $j$. The other legs attached to the vertex $A$ are decorated with spins $j_{1},..,j_{P}$, while the other edges attached to $B$ carry the spins $k_{1},..,k_{Q}$. The Hilbert space of (bounded) spin network states for this configuration at fixed spins is:
\be
\cH_{AB}^{0}=\cH^0_{A}\otimes\cH^0_{B}\,,
\ee
where $\cH^0_{A}$ and $\cH^0_{B}$ are the spaces of intertwiners attached to the two nodes:
\be
\left\{
\begin{array}{l}
\cH^0_{A}=\mathrm{Inv}_{\SU(2)}
\big{[}
V^{j_{1}}\otimes..\otimes V^{j_{P}}
\otimes V^{j}
\big{]}
\vspace*{2mm}\\
\cH^0_{B}=\mathrm{Inv}_{\SU(2)}
\big{[}
V^{k_{1}}\otimes..\otimes V^{k_{Q}}
\otimes V^{j}
\big{]}
\end{array}
\right.
\ee
The subscript $0$ is to remind that all these states are $\SU(2)$-invariant.
Considering a (pure) state $|\psi\ra\in\cH_{AB}^{0}$, we define the entanglement between $A$ and $B$ as the Von Neumann entropy of the reduced density matrices:
\be
\rho=|\psi\ra\la\psi|
\,,\quad
\rho_{A}=\tr_{B}\rho
\,,\quad
\rho_{B}=\tr_{A}\rho
\,,\nn
\ee
\be
E(A|B)=-\tr\,\rho_{A}\ln\rho_{A}=-\tr\,\rho_{B}\ln\rho_{B}\,.
\ee
This entanglement entropy is bounded by the dimension of the intertwiner spaces $\cH^0_{A}$ and $\cH^0_{B}$:
\be
E(A|B)\le \min\big{(}
\ln \dim\cH^0_{A},
\ln \dim\cH^0_{B}
\big{)}\,.
\ee
\begin{figure}[h]

\begin{tikzpicture}[scale=1]

\coordinate(a) at (0,0) ;
\coordinate(b) at (1.5,0);

\draw (a) node {$\bullet$} ++(0.2,0.3) node{$I_{A}$};
\draw (b) node {$\bullet$} ++(-0.2,0.3)  node{$I_{B}$};
\alink{a}{b}{below}{j};

\draw (a)--++(-.85,0.3) ;
\draw (a)--++(-.8,0.7);
\draw (a)--++(-.8,-.7);

\draw (b)--++(.85,0.3) ;
\draw (b)--++(.8,0.7) ;
\draw (b)--++(.8,-.7)  ;

%

\draw[dotted,line width=.7pt] (-.8,0) -- (-.8,-.4);
\draw[dotted,line width=.7pt] (1.5+.8,0) -- (1.5+.8,-.4);

\draw[dashed] (0,0.1) ellipse (16pt and 17pt);
\draw[dashed] (1.5,0.1) ellipse (16pt and 17pt);
\draw[<->] (.2,.8)  to[bend left] node[midway,above]{$E(A|B)$} (1.3,.8);

\end{tikzpicture}
\vspace*{4mm}

\begin{tikzpicture}[scale=1]

\coordinate(a) at (0,0) ;
\coordinate(b) at (1.5,0);

\draw (a) node {$\bullet$};
\draw (b) node {$\bullet$};
\alink{a}{b}{above}{j};

\draw (a)--++(-.85,0.3) node[left]{$j_{2}$};
\draw (a)--++(-.8,0.7) node[left]{$j_{1}$};
\draw (a)--++(-.8,-.7)  node[left]{$j_{P}$};

\draw (b)--++(.85,0.3) node[right]{$k_{2}$};
\draw (b)--++(.8,0.7) node[right]{$k_{1}$};
\draw (b)--++(.8,-.7)  node[right]{$k_{Q}$};

\draw[dotted,line width=1pt] (-.95,0) -- (-.95,-.4);
\draw[dotted,line width=1pt] (1.5+.95,0) -- (1.5+.95,-.4);

\draw[dashed] (-1.05,0) ellipse (16pt and 30pt);
\draw[dashed] (2.55,0) ellipse (16pt and 30pt);
\draw[<->] (-.7,-1)  to[bend right] node[midway,above]{$\tE(A|B)$} (2.2,-1);

\end{tikzpicture}

\caption{Intertwiner entanglement entropy (above) versus boundary spin state entanglement (below): one can either look directly at the entanglement $E(A|B)$ between the intertwiner states at the two nodes or probe the correlation between the two nodes by looking at the entanglement $\tE(A|B)$  it induces on the boundary spins.
\label{fig:entanglement}}

\end{figure}
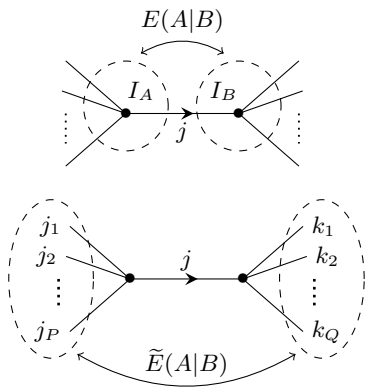

The other way to look at the regions $A$ and $B$ in to glue them along the intermediate edge and group them into a single region $AB$. The edge between $A$ and $B$, carrying the spin $j$, is now within this coarser region $AB$, while the remaining edges carrying the spins $j_{1},..,j_{P},k_{1},..,k_{Q}$ are all on the boundary. Thus the Hilbert space of boundary spin states is:
\be
H^{\pp}_{AB}=
\big{(}
V^{j_{1}}\otimes..\otimes V^{j_{P}}
\big{)}
\otimes
\big{(}
V^{k_{1}}\otimes..\otimes V^{k_{Q}}
\big{)}\,,
\ee
which is clearly not the same Hilbert space as $\cH^{0}_{AB}$.
Considering a basis state $I_{A}\otimes I_{B}$ in $\cH^{0}_{AB}$, which is a straightforward tensor product of two intertwiners living respectively at the nodes $A$ and $B$, we can project it onto a boundary state in $H^{\pp}_{AB}$ by gluing the two intertwiners together:
\begin{align}
&\cP:\,I_{A}\otimes I_{B}\in\cH^{0}_{AB} \\
\mapsto\,
&\f1{\sqrt{2j+1}}\sum_{m=-j}^{j}
(-1)^{j-m}
\la j,m| I_{A} \otimes \la j,-m|I_{B}
\,\,\in H^{\pp}_{AB} \nn
\end{align}
where we have inserted the (trivial) bivalent intertwiner to glue back the spin $j$ part of the intertwiner $I_{A}$ with its counterpart in the intertwiner $I_{B}$. Then we extend this map linearly to the whole Hilbert space $\cH^{0}_{AB}$.

Loop quantum gravity actually defines a whole class of such maps by allowing for a $\SU(2)$ group element -holonomy- among the intermediate edge:
\begin{align}
&\cP_{g}:\,I_{A}\otimes I_{B}\in\cH^{0}_{AB} \\
\mapsto\,
&\f1{\sqrt{2j+1}}\sum_{a,b}
(-1)^{j-a} D^{j}_{ab}(g)
\la j,-a| I_{A} \otimes \la j,b|I_{B}
\,\in H^{\pp}_{AB}
\nn
\end{align}
where we have inserted the Wigner matrix $D^{j}_{ab}(g)$ representing the group element $g\in\SU(2)$ in the spin-$j$. The case $g=\id$ reproduces the previous map, $\cP_{\id}=\cP$. This case is special since it actually produces an intertwiner between the boundary spins:
\be
\cP\,[I_{A}\otimes I_{B}]
\,\in\mathrm{Inv}_{\SU(2)}\big{[}
H^{\pp}_{AB}
\big{]}\,.
\ee
For any other group element, this is not the case anymore\footnotemark{} and we reach arbitrary states in the boundary Hilbert space $H^{\pp}_{AB}$.
\footnotetext{It is nevertheless possible to get back an intertwiner by integrating the $\cP_{g}$ maps over equivalence classes of $\SU(2)$ group elements under conjugation \cite{Livine:2006xk}:
$$
\int \rd g\,
\delta_{\theta}(g)
\,\cP_{g}\,[I_{A}\otimes I_{B}]
\in\mathrm{Inv}_{\SU(2)}\big{[}
H^{\pp}_{AB}
\big{]}\,,
$$
where we integrate over all group elements, whose rotation angle is $\theta$. However, this averaging procedure does not have a clear physical or geometrical meaning for spin networks.
}

Focusing on the projection map $\cP=\cP_{\id}$, we can now define the spin state entanglement, considering the natural bipartite splitting of the boundary Hilbert space:
\be
H^{\pp}_{AB}=H_{A}^{\setminus (AB)}\otimes H_{B}^{\setminus (AB)}\,,
\ee
\be
H_{A}^{\setminus (AB)}=V^{j_{1}}\otimes..\otimes V^{j_{P}}
\,,\,\,
H_{B}^{\setminus (AB)}=V^{k_{1}}\otimes..\otimes V^{k_{Q}}\,.
\nn
\ee
This leads to defining a notion of spin state entanglement, by considering the density matrix on $H^{\pp}_{AB}$. Starting with a (pure) state $|\psi\ra$ in $\cH^{0}_{AB}$, we glue the intertwiners using the map $\cP$, get a state in the boundary Hilbert space $H^{\pp}_{AB}$, compute the reduced density matrices and finally the resulting entanglement entropy:
\be
\tpsi= \cP[\psi]
\,,\,\,
\trho=|\tpsi\ra\la \tpsi| \,\in \mathrm{End}[H^{\pp}_{AB}]\,,
\ee
\be
\trho_{A}=\tr_{B}\trho
\,,\quad
\tE(A|B)
=-\tr \trho_{A}\ln\trho_{A}
\,.
\ee
%
This is straightforwardly extended to any group element insertion along the intermediate edge, leading to defining spin state  entanglements $\tE_{g}(A|B)$ for any $g\in\SU(2)$.
These are all bounded by the dimensions of the boundary Hilbert spaces:
\be
\tE_{g}(A|B)\le \min(\ln\dim H_{A}^{\setminus (AB)},\,\ln\dim H_{B}^{\setminus (AB)})\,,
\ee
\be
\dim H_{A}^{\setminus (AB)}=\prod_{i=1}^{P}(2j_{i}+1)
\,,\,
\dim H_{B}^{\setminus (AB)}=\prod_{i=1}^{Q}(2k_{i}+1)\,.
\nn
\ee
For any set of spins, these are always (much) larger than the intertwiner space dimensions \footnotemark.
\footnotetext{
Simply due to the fact that all tensor products of spins are entirely reducible, one can recouple any set of spins to a single global spin:
$$
\bigotimes_{i=1}^{P}V^{j_{i}}=\bigoplus_{J}^{\sum_{i}j_{i}}
\mathrm{Inv}_{\SU(2)}\Big{[}V^{J}\otimes\bigotimes_{i=1}^{P}V^{j_{i}}\Big{]}\,.
$$
Taking  $J=j$, this means that $\ln \dim\cH^0_{A}\le\dim H_{A}^{\setminus (AB)}$.
}

As graphically represented on fig.\ref{fig:entanglement}, the intertwiner entanglement is a priori {\it not equal} to the spin state entanglement. The first deals entirely with $\SU(2)$-invariant states, while the latter is based a gluing procedure (which can be viewed as a coarse-graining map) and produces non-necessarily-gauge-invariant states. Our goal is to compare this two notions.
We will show below that:
\begin{itemize}
\item For a pure basis state, $I_{A}\otimes I_{B}$, the intertwiner entanglement obviously vanishes, $E=0$, while the spin state entanglement does not and is actually given exactly by $\tE=\ln(2j+1)$.

\item For an arbitrary pure state in $\cH^{0}_{AB}$, but not necessarily a straightforward tensor product state, both entanglements $E$ and $\tE$ do not vanish. Nevertheless, the difference between these two notions is always the same, $\Delta E=\tE-E=\ln(2j+1)$.

\item Holonomy insertions do not  matter, they do not affect the spin network entanglement,  $\tE_{g}(A|B)=\tE(A|B)$ for all group elements $g\in\SU(2)$ along the edge $A$-$B$ linking the two intertwiners.

\end{itemize}



\begin{prop}
\label{prop:superposition}
{\bf Intertwiner superpositions:}
\vspace*{1mm}\\
Let $\cI\in\cH^0_{AB}=\cH^0_{A}\otimes \cH^0_{B}$ be an arbitrary normalized intertwiner,  i.e. an arbitrary superposition of tensor product states, potentially carrying a non-trivial entanglement between its two parts $A$ and $B$. We consider its boundary state $\cP\,[I_{A}\otimes I_{B}] \in H^{\pp}_{AB}$, obtained by gluing the two intertwiners along their common spin-$j$ edge. Then the difference between the boundary spin entanglement, for the state $\cP\,[I_{A}\otimes I_{B}]$, and the intertwiner entanglement, for the state $I_{A}\otimes I_{B}$, only depends on the spin-$j$:
\be
\tE(A|B)=E(A|B)+\ln(2j+1)\,.
\ee

\end{prop}
\begin{proof}
We choose orthonormal basis for the intertwiners at the two nodes $A$ and $B$:
\be
I_{A}^k\in\cH^0_{A}=\mathrm{Inv}\big{[}
\bigotimes_{i=1}^pV^{j_{i}}\otimes V^j
\big{]}
\,,
\,\,
\la I_{A}^k|I_{A}^{\tk}\ra=\delta_{k\tk}\,,
\ee
and the same for the intertwiner space $\cH^0_{B}$. We consider an arbitrary state:
\be
\cI=\sum_{kl}\psi_{kl}\,|I_{A}^k\ra\otimes |I_{B}^l\ra
\,,\quad
\la \cI|\cI\ra=\tr \psi\psi^\dagger=1
\,.
\ee
The matrix $(\psi)_{kl}$ is a priori not a square matrix, but the matrix $\psi\psi^\dagger$ is always a square matrix. We compute the reduced density matrix:
\be
\rho_{A}=\sum_{k,\tk}
\big{(}\psi\psi^\dagger\big{)}_{k\tk}\,
|I_{A}^k\ra\la I_{A}^{\tk}|
\,,
\quad
\tr\rho_{A}=1
\,,
\ee
from which we deduce the intertwiner entanglement:
\be
E(A|B)=-\tr\,\psi\psi^\dagger\ln\,\psi\psi^\dagger
\,.
\ee
Now we glue the intertwiners at nodes $A$ and $B$ together along the link $A$-$B$ using the bivalent intertwiner for the spin $j$, i.e. the singlet state in $V^j\otimes V^j$:
\be
\cP[\cI]=
\f1{\sqrt{2j+1}}\sum_{kl}\psi_{kl}\,
\sum_{m=-j}^{j}(-1)^{j-m}
|I_{A}^{k,m}\ra\otimes|I_{B}^{l,-m}\ra
\,,
\nn
\ee
with the notation $|I_{A}^{k,m}\ra=\la j,m | I_{A}^k\ra\,\in\bigotimes_{i=1}^pV^{j_{i}}$. The crucial property of intertwiners is these states are orthonormal:
\be
\la I_{A}^{k,m} | I_{A}^{\tk,\tm}\ra
=\,
\f1{2j+1}\,
\delta_{k,\tk}\delta_{m,\tm}
\,.
\ee
Indeed, each intertwiner basis label $k$ defines an embedding of the tensor product $\bigotimes_{i=1}^pV^{j_{i}}$ into the single spin space $V^j$, i.e. a channel recoupling the spins $j_{1},..,j_{P}$ into the spin $j$. This implies that projecting the intertwiner\footnotemark{} on different magnetic moments $m$ lead to orthogonal states in $V^j\hookrightarrow\bigotimes_{i=1}^pV^{j_{i}}$.
\footnotetext{
Equivalently, the original intertwiner $I_{A}^k$ can be recovered from its projections:
\be
|I_{A}^k\ra=\sum_{m=-j}^j |I_{A}^{k,m}\ra\otimes |j,m\ra.
\nn
\ee
And we recover the norm $\la I_{A}^k|I_{A}^k\ra = \sum_{m}\la I_{A}^{k,m}|I_{A}^{k,m}\ra=1$.
}
The same holds for the node $B$.

This allows to compute the norm of the boundary state $\cP[\cI]$:
\be
\la\, \cP[\cI]\,|\,\cP[\cI]\,\ra\,=\,\f1{(2j+1)^2}
\,,
\ee
which means that we need to renormalize the state and  consider the normalized state $|\cP[\cI]\ra_{n}=(2j+1)\,|\cP[\cI]\ra$.
%
%
Using this correct normalization, we compute the reduced density on the boundary spins of $A$, as an endomorphism on $H_{A}^{\setminus (AB)}=\bigotimes_{i=1}^pV^{j_{i}}$:
\be
\trho_{A}=
\sum_{k\tk}(\psi\psi^\dagger)_{k\tk}\sum_{m=-j}^j
|I_{A}^{k,m}\ra\la I_{A}^{\tk,m}|
\,,\quad
\tr\,\trho_{A}=1
\,.
\ee
This directly gives the spin state entanglement:
\be
\tE(A|B)=\ln(2j+1)-\tr\,\psi\psi^\dagger\ln\,\psi\psi^\dagger
\,,
\ee
thus concluding the proof of the proposition.

\end{proof}

As a consequence of this proposition, we apply it to pure tensor product states, which do not carry any intertwiner entanglement, $E(A|B)=0$, thus getting a spin state entanglement $\tE(A|B)$ on the boundary only depending on the spin $j$ carried by the link between $A$ and $B$
\begin{coro}
{\bf Basis States:}
\vspace*{1mm}\\
Let $I_{A}\otimes I_{B}\in\cH^0_{AB}=\cH^0_{A}\otimes \cH^0_{B}$ be a normalized intertwiner tensor product state. We consider its boundary state $\cP\,[I_{A}\otimes I_{B}] \in H^{\pp}_{AB}$, obtained by gluing the two intertwiners along their common spin-$j$ edge using the bivalent intertwiner (or singlet state). Then the boundary spin entanglement carried by this tensor product intertwiner is simply:
\be
\tE(A|B)=\ln(2j+1)\,.
\ee
\end{coro}


We can generalize these results on entanglement to the case of a non-trivial $\SU(2)$ holonomy along the edge linking the two intertwiners $A$ and $B$. We start by the entanglement between the two spin states living on the edge itself.

\begin{lemma}
{\bf Bivalent Intertwiners:}
\vspace*{1mm}\\
Let us consider the singlet state of two spins $j$:
\be
\sum_{m} (-1)^{j-m}\,|j,m\ra\otimes  |j,-m\ra\
\in\,\mathrm{Inv}_{\SU(2)}\big{[}
V^j\otimes V^j\big{]}
\,.
\nn
\ee
The entanglement between the two spins is maximal and its value is $\ln(2j+1)$. If the two spins are related by a $\SU(2)$ group element $g\in\SU(2)$,
\be
\sum_{a,b} (-1)^{j-b}D^j_{ab}(g)\,|j,a\ra\otimes  |j,-b\ra
\,\in\, 
V^j\otimes V^j
\,,
\nn
\ee
then this amounts to a local unitary transformation on one of the two spins and thus does not affect the entanglement: the two spins remain maximally entangled.

\end{lemma}
\begin{proof}
The proof is straightforward, since the Wigner matrix $D^j(g)$ representing the group element $g\in\SU(2)$ is a unitary matrix.
\end{proof}

Generalizing this lemma to arbitrary $N$-valent intertwiners, we can extend the proposition \ref{prop:superposition} to gluing of the intertwiners using an arbitrary $\SU(2)$ holonomy along the  shared edge $A$-$B$:

\begin{prop}
{\bf Holonomy insertions:}
\vspace*{1mm}\\
Let $\cI\in\cH^0_{AB}=\cH^0_{A}\otimes \cH^0_{B}$ be an arbitrary normalized intertwiner. We glue the two intertwiners by the shared edge $A$-$B$ carrying the fixed spin $j$ and an arbitrary group element $g\in\SU(2)$, thus defining the boundary state $\cP_{g}\,[I_{A}\otimes I_{B}] \in H^{\pp}_{AB}$. Then the boundary spin state entanglement $\tE(A|B)$ does not depend on the insertion of $g$ and the difference between the  spin state entanglement and the intertwiner entanglement remains $ \tE(A|B)-E(A|B)=\ln(2j+1)$, as in the $g=\id$ case with a trivial gluing of the two intertwiners.

\end{prop}
\begin{proof}
We prove this statement by follow the exact same steps as for the proof of proposition \ref{prop:superposition}. At the final stage, when computing the reduced density matrix for the bipartite splitting $A|B$ of the boundary spins, we realize that the group element $g\in\SU(2)$ disappears since it defines a local unitary transformation acting only on $B$ (or acting only on $A$) as for the lemma above.
\end{proof}

In fact, one can always change the $\SU(2)$ holonomy along the edge linking the two nodes $A$ and $B$, and in particular set $g_{AB}=\id$, by doing an arbitrary $\SU(2)$ gauge transformation at $A$ or at $B$. This is actually at the heart of the ``coarse-graining by gauge-fixing'' procedure introduced in \cite{Livine:2006xk,Livine:2013gna,Charles:2016xwc}. So it is completely natural that the entanglement between the two nodes $A$ and $B$ does not depend on that holonomy. On the other hand, if we were considering three nodes, $A$, $B$ and $C$, all linked to each other by spin network edges, one could not gauge-away the $\SU(2)$ holonomy around the closed loop $A\rightarrow B\rightarrow C\rightarrow A$ and the tri-partite entanglement between the three intertwiners would most certainly depend non-trivially on this holonomy. Also, if we were to consider a superposition of $\SU(2)$ group elements $g$, or a superposition of spins $j$, on the edge linking $A$ and $B$, this would most necessarily affect the entanglement.

\section{Spin-$\f12$ examples and Qubit entanglement}

To illustrate the distinction between the intertwiner entanglement and the spin state entanglement, let us give examples with 3-valent and 4-valent intertwiners.

Let us start with two 3-valent intertwiners, as depicted on fig.\ref{fig:3valent}. On one side, we are intertwining the spins $j_{1}$ and $j_{2}$ with the common spin $j$, while we are intertwining the spins $k_{1}$ and $k_{2}$ with $j$ on the other side. Assuming that the triangular inequalities are satisfied at both nodes, the 3-valent intertwiners are unique so that they can not be any intertwiner entanglement. We nevertheless get a non-vanishing spin state entanglement due to the spin $j$ channel on the link:
\be
E(A|B)=0
\,,\quad
\tE(A|B)=\ln(2j+1)
\,.
\ee
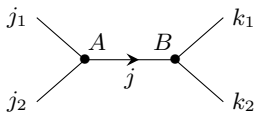
\begin{figure}[h]

\begin{tikzpicture}[scale=.8]

\coordinate(a) at (0,0) ;
\coordinate(b) at (1.5,0);

\draw (a) node {$\bullet$} ++(0.2,0.3) node{$A$};
\draw (b) node {$\bullet$} ++(-0.2,0.3)  node{$B$};
\alink{a}{b}{below}{j};

\draw (a)--++(-.8,0.7) node[left]{$j_{1}$};
\draw (a)--++(-.8,-.7)  node[left]{$j_{2}$};

\draw (b)--++(.8,0.7) node[right]{$k_{1}$};
\draw (b)--++(.8,-.7)  node[right]{$k_{2}$};

\end{tikzpicture}

\caption{The case of two 3-valent nodes is simple: 3-valent intertwiners, between three given spins, are unique, so they can not be any intertwiner entanglement, $E(A|B)=0$, and the spin state entanglement is entirely due to the spin $j$ channel linking the two nodes, $\tE(A|B)=\ln (2j+1)$.
\label{fig:3valent}}

\end{figure}

Let us describe the first non-trivial example using 4-valent intertwiners between spin $\f12$, as on fig.\ref{fig:4valent}. In this case, both intertwiner spaces at the nodes $A$ and $B$ are two-dimensional:
\be
\dim\,\cH^0_{A}
\,=\,
\dim\,\cH^0_{B}
\,=2
\,.
\ee
So we have one 2-level system, or qubit, at each node. This is the basic mapping of spin network dynamics onto qubit systems \cite{Feller:2015yta}. One can choose a recoupling basis for the intertwiners. For instance, we can pair together the spins $j_{1}$ with $j_{2}$, and $j_{3}$ with $j$, recoupling each pair to the intermediate spin 0 or 1, thus defining two intertwiner basis states $|0\ra_{A}$ and $|1\ra_{A}$. One can also choose another pairing, or an orthonormal basis diagonalizing an interesting Hermitian observable such as the volume operator eigenstates as in \cite{Feller:2015yta}. The choice of intertwiner basis states does not matter for the entanglement in the end.
\begin{figure}[h]
\begin{center}

\begin{tikzpicture}[scale=1.2]

\coordinate(a) at (0,0) ;
\coordinate(a0) at (-.5,.6) ;
\coordinate(b) at (1.5,0);
\coordinate(b0) at (2,.6) ;

\draw (a) node {$\bullet$}   ++(0,-.2) node{$A$};
\draw (b) node {$\bullet$}  ++(-.1,-.2) node{$B$};
\alink{a}{b}{above}{j=\f12};

\draw[dashed, thick] (a0)--(a) ++(-.75,0.2) node{0 or 1};
\draw (a0)--++(-.7,0.6) node[left]{$\f12=j_{1}$};
\draw (a0)--++(-.75,0.1) node[left]{$\f12=j_{2}$};
\draw (a)--++(-.8,-.7)  node[left]{$\f12=j_{3}$};

\draw[dashed, thick] (b0)--(b) ++(0.7,0.2) node{0 or 1};
\draw (b0)--++(.7,0.6) node[right]{$k_{1}=\f12$};
\draw (b0)--++(.75,0.1) node[right]{$k_{2}=\f12$};
\draw (b)--++(.8,-.7)  node[right]{$k_{3}=\f12$};

\end{tikzpicture}

\caption{The simplest case with intertwiner entanglement is to consider two 4-valent intertwiners between spins $\f12$. The intertwiners define a two-level system at each node, for example by recoupling at the node $A$ the spins $j_{1}$ and $j_{2}$ to a spin 0 or 1, and similarly at the node $B$. This allows to consider tensor product states such as $|0\ra_{A}\otimes |0\ra_{B}$ or $|1\ra_{A}\otimes |1\ra_{B}$, as well as Bell states maximally entangling the intertwiners at the two nodes. In all cases, we can compute the intertwiner entanglement $E(A|B)$ and the boundary spin state entanglement and we always find that their difference is given by the spin on the edge linking the two nodes, $\ln(2j+1)=\ln 2$.
\label{fig:4valent}}
\end{center}
\end{figure}
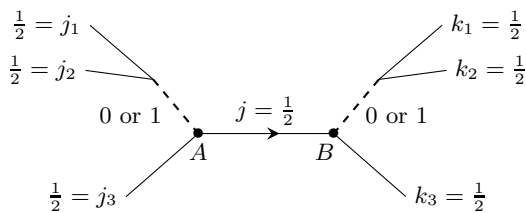
Calculations writing explicitly the detailed state of the six boundary spins (in terms of up and down states) can be lengthy but are nevertheless straightforward. On the one hand, for a tensor product state, such as $|0\ra_{A}\otimes |0\ra_{B}$, one finds that the intertwiner entanglement vanishes while the spin state entanglement entropy sees a maximally entangled pair along the edge $A$-$B$:
\be
\cI=|0\ra_{A}\otimes |0\ra_{B}
\quad\Longrightarrow \quad
\left|
\begin{array}{l}
E(A|B)=0
\\
\tE(A|B)=\ln 2
\end{array}
\right.
\ee
On the other hand, for a Bell state, such as $(|0\ra_{A}\otimes |0\ra_{B}-|1\ra_{A}\otimes |1\ra_{B})/\sqrt{2}$, the intertwiners are already maximally entangled and the spin state entanglement doubles this, thereby respecting the general analysis presented in the previous section:
\be
\cI=
\f{|0\ra_{A}\otimes |0\ra_{B}-|1\ra_{A}\otimes |1\ra_{B}}{\sqrt{2}}
\,\Longrightarrow \,
\left|
\begin{array}{l}
E(A|B)=\ln 2
\\
\tE(A|B)=2\ln 2
\end{array}
\right.
\nn
\ee

More generally, it might be interesting to develop further the mapping of the spin-$\f12$ sector or the spin network Hilbert space onto qubit systems as used in condensed matter models, in order to study further entanglement properties of spin network states, for instance entanglement between far nodes and multi-partite entanglement.

\section{Outlook}

To summarize, we consider two nearest neighbor vertices, $A$ and $B$, on a spin network states, linked with an edge of the graph, and we consider the region of space formed by putting the two vertices together: the vertices and the intermediate edge form the bulk , while all the other edges attached to those vertices form the region's boundary. We look at the spin state entanglement entropy $\tE(A|B)$ between the boundary spins attached to $A$ and those attached to $B$. Assuming that the intermediate edge carries a definite spin $j$ (i.e. not a superposition), the spin state entanglement is the sum of two terms, the entanglement $E(A|B)$ between the intertwiners living at the two vertices plus a fixed contribution $\ln(2j+1)$ from the gluing of the two intertwiners.
This is very similar to the results obtained in \cite{Delcamp:2016eya} using the more involved technology of spin tubes and the fusion basis of lattice gauge field theories.

They are direct extensions of the results presented here, studying superpositions of the spin living on the edge linking the two nodes or inserting a $\SL(2,\C)$ group element or little loops along that edge to allow for local curvature excitations \cite{Charles:2016xwc,Loops17:Livine}. Besides these short-term improvements, our simple analysis raises a certain number of potentially deeper questions:

\begin{itemize}

\item The notion of entanglement that one should use depends on the observable that we would like to measure. For instance, correlations on a boundary would require computing the spin state entanglement, while correlations between volume excitations would involve the intertwiner entanglement. But this begs to clarify the role of the graph underlying the spin network states. The standard view in loop quantum gravity is that the graph is the backbone of the discrete quantum geometries defined by the spin network states. However, the entanglement between two intertwiners does not depend on the link or gluing process along that link, it directly measure the correlations between the gauge-invariant degrees of freedom sitting at the  vertices. In particular, computing the intertwiner entanglement between two nodes does not depend on whether these nodes are nearest neighbors or related to a long chain of spin network links.

So what is the role of the spin network edges (beside mathematically setting constraints of matching spin at both ends)? Actually spin network edges might not be physically important information, despite traditionally being considered  as a central ingredient in loop quantum gravity since the edges carry the spins which are the quanta of area. Indeed the action of diffeomorphisms (generated by the Hamiltonian constraints) on spin networks is usually understood as modifying the graph structure and combinatorics. From this perspective, the graph can be interpreted as describing the sampling of the space geometry according to one observer.  Then another observer could see a different (spin) network but defining the  same structure of correlations between points (identified here as the spin network vertices). So intertwiner entanglement would become the fundamental notion. This is a similar point of view as in group field theory \cite{Oriti:2014uga,Chirco:2017vhs}.

\item Since the spin state entanglement on the boundary sees the spin $j$ on the intermediate link between the two nodes, this suggests to use this entanglement formula, and potentially its generalization to many-node entanglement, to probe the bulk structure. For instance, if we only had access to the boundary spin, as coarse-graining scenarios, we could explore all bipartite splitting of the boundary and look for a possible reconstruction of the bulk in terms of two spin network vertices. 

\item The distinction between intertwiner entanglement and spin state entanglement invites to revisit the use of entanglement entropy for black holes in loop quantum gravity \cite{Donnelly:2008vx,Perez:2014ura}. Indeed, the standard point of view is that the black hole boundary cuts spin network edges so that the horizon is a collection of spin states -$\SU(2)$ punctures- defining a quantum surface \cite{Asin:2014gta,Feller:2017ejs}. So as explained in \cite{Perez:2014ura}, all the entanglement entropy between the black hole interior and exterior is usually taken as the sum of the $\ln(2j+1)$ contributions coming from all the edges crossing the horizon. However, the present work offers the perspective that the $\ln(2j+1)$ term, coming from the gluing by a bivalent intertwiner on each edge, is not the main contribution of the entanglement entropy on a spin network. So we would like to challenge the current framework for black hole entropy in loop quantum gravity and suggest that it might be more physically relevant to look at entanglement between intertwiners - the actual geometry excitations- inside and outside the horizon.

One can further argue that the $\ln(2j+1)$ contributions come from breaking the $\SU(2)$ gauge invariance at the boundary. To settle this issue and cleanly defining the relevant entanglement entropy for black holes in loop quantum gravity, one way out could be to introduce a concept of gauge-invariant boundary, for instance as a layer of loops consisting in the edges crossed by the horizon together with their end vertices and a minimal set of spin network edges connecting all the boundary edges, as suggested in \cite{Feller:2017ejs}. One would then be looking at the three-body entanglement and correlation between the boundary layer, the bulk geometry within the black hole and the exterior geometry.

\end{itemize}

To conclude, it is a whole research program to analyze the entanglement structure of spin networks and look for ``good'' spin network states with good correlation and entanglement, thereby allowing for a semi-classical interpretation.  Indeed the hope is to use  correlations to reconstruct a notion of distance between parts of the spin networks. We would  then study the multi-body entanglement between three or more regions of spin networks in order to discuss curvature. More generally, it would probably be very interesting to study the mutual information between arbitrary subsystems of spin networks, following \cite{e17053253}, and map the resulting hierarchy of correlations to a hierarchy of geometrical observables on spin networks.
This line of thoughts seems to converge with the use of the multiscale entanglement renormalization  ansatz (MERA from tensor network renormalization in condensed matter) to the holographic reconstruction of the bulk geometry in the framework of the AdS/CFT correspondence \cite{Hu:2017rsp,Hayden,Yoshida}. The goal would be to identify the equivalent of local entangling and disentangling operations for spin network states in order to build semi-classical spin network states with an interesting entanglement structure and respecting a local version of holography.

\section*{Acknowledgement}

I would like to thank the organizers of the TGSI 2017 conference at the CIRM (Marseille, France), Aug 28 - Sept 1 2017, for the great interdisciplinary and collaborative atmosphere, which motivated me to explore  the information point of view on loop quantum gravity.

I am also grateful to Simone Speziale and Goffredo Chirco for stimulating discussions on the structure of spin networks and encouragements on clarifying the Hilbert spaces used in entanglement calculations.



\newpage

\bibliographystyle{bib-style}
\bibliography{LQG}

\end{document}